\documentclass[reqno,10pt,a4paper]{amsart}

\usepackage{amssymb,mathptmx,eucal,array,setspace,color,geometry,enumitem,cite}
\usepackage{graphicx}
\usepackage{tikz}

\DeclareSymbolFont{largesymbols}{OMX}{zplm}{m}{n} 

\numberwithin{equation}{section}

\newcolumntype{C}{>{$}c<{$}} 

\allowdisplaybreaks

\newcommand{\alg}[1]{\mathfrak{#1}} 

\newcommand{\tfunc}[2]{#1 \bigl( #2 \bigr)} 

\newcommand{\brac}[1]{\left( #1 \right)}

\newcommand{\dd}{\mathrm{d}}   
\newcommand{\ii}{\mathfrak{i}} 

\newcommand{\wun}{\mathbf{1}}  

\newcommand{\normord}[1]{\mbox{${} : #1 : {}$}} 

\newcommand{\bra}[1]{\bigl\langle #1 \bigr\rvert}
\newcommand{\ket}[1]{\bigl\lvert #1 \bigr\rangle}
\newcommand{\braket}[2]{\bigl\langle #1 \bigr\rvert \bigl. #2 \bigr\rangle}          
\newcommand{\bracket}[3]{\bigl\langle #1 \bigr\rvert #2 \bigl\lvert #3 \bigr\rangle} 

\newcommand{\affine}[1]{\widehat{#1}}
\newcommand{\AKMA}[2]{\affine{\alg{#1}} \left( #2 \right)}                  

\newcommand{\MinMod}[2]{\mathsf{M} \bigl( #1 , #2 \bigr)}                   
\newcommand{\MinIdeal}[2]{\mathsf{I} \bigl( #1 , #2 \bigr)}                  

\newcommand{\FF}[1]{\mathcal{F}_{#1}}          

\newcommand{\scr}[1]{\mathcal{Q}_{#1}}
\newcommand{\scrs}[2]{\mathcal{Q}_{#1}^{[#2]}}

\newcommand{\monsym}[1]{\mathsf{m}_{#1}}                  
\newcommand{\fmonsym}[2]{\tfunc{\mathsf{m}_{#1}}{#2}}     
\newcommand{\elsym}[1]{\mathsf{e}_{#1}}                   
\newcommand{\felsym}[2]{\tfunc{\mathsf{e}_{#1}}{#2}}

\newcommand{\powsum}[1]{\mathsf{p}_{#1}}                  
\newcommand{\fpowsum}[2]{\tfunc{\mathsf{p}_{#1}}{#2}}
\newcommand{\fjack}[3]{\tfunc{\mathsf{P}_{#1}^{#2}}{#3}}

\newcommand{\jprod}[3]{\left\langle #1 \right\rangle^{#2}_{#3}}

\newcommand{\cft}{conformal field theory}
\newcommand{\cfts}{conformal field theories}

\newcommand{\voa}{vertex operator algebra}
\newcommand{\voas}{vertex operator algebras}

\theoremstyle{plain}
\newtheorem{thm}{Theorem}
\newtheorem{prop}[thm]{Proposition}

\newtheorem{cor}[thm]{Corollary}

\begin{document}

\title{From Jack polynomials to minimal model spectra}

\author[D Ridout]{David Ridout}

\address[David Ridout]{
Department of Theoretical Physics \\
Research School of Physics and Engineering;
and
Mathematical Sciences Institute;
Australian National University \\
Acton, ACT 2600 \\
Australia
}

\email{david.ridout@anu.edu.au}

\author[S Wood]{Simon Wood}

\address[Simon Wood]{
Department of Theoretical Physics \\
Research School of Physics and Engineering;
and
Mathematical Sciences Institute;
Australian National University \\
Acton, ACT 2600 \\
Australia
}

\email{simon.wood@anu.edu.au}

\begin{abstract}
  In this note, a deep connection between free field realisations of \cfts{}
  and symmetric polynomials is presented. We give a brief introduction into
  the necessary prerequisites of both free field realisations and symmetric
  polynomials, in particular Jack symmetric polynomials. Then we combine these
  two fields to classify the irreducible representations of the minimal model \voas{} as an
  illuminating example of the power of these methods. While these results
  on the representation theory of the minimal models are all known, 
  this note exploits the full power of Jack polynomials to present significant
  simplifications of the original proofs in the literature.
\end{abstract}

\maketitle

\onehalfspacing

\section{Introduction}

Free field theories have long been of great interest to \cft. Not only are
they elegant tractable \cfts{} in their own right, but they are also a versatile
tool for realising more complicated \cfts{} and making them tractable. The purpose
of this note is to present, in a simple and familiar setting, a deep
connection between free field theories and Jack symmetric polynomials. The symmetric
polynomial methods will then be applied to the well known free field realisations
of the Virasoro minimal models. However, it is
important to stress that these methods work far more generally. We have simply
chosen to discuss applications to the minimal models for pedagogical purposes.

A different example, where Jack symmetric polynomials have recently garnered a
lot of attention, is the much
celebrated AGT conjecture \cite{AGT10}, which relates \cfts{} to the instanton
calculus of Yang-Mills theories. The appearance of symmetric polynomials is
due to the \cfts{} in question being resolved by Coulomb gas free field
theories, just as the \cfts{} in this note are. Contrary to what was initially
believed, it does not seem that Jack symmetric polynomials form the most
natural basis for understanding the AGT conjecture. Rather, a generalisation
of Jack symmetric polynomials seems to be needed \cite{MorSmiAGT14}.

The two main results discussed in this note are Theorems
\ref{sec:singvecformula} and \ref{sec:specthm}. Theorem
\ref{sec:singvecformula}, which is originally due to Mimachi and Yamada
\cite{MimSing95}, gives elegant formulae for Virasoro singular vectors in Fock
modules, while Theorem \ref{sec:specthm}, which is originally due to Wang
\cite{WangRat93}, determines the conformal highest weights of
the irreducible representations of the minimal model vertex operator algebras (also called chiral algebras).
The original proofs, impressive though they are, are rather complicated and
this note gives novel, drastically shortened and streamlined proofs by using
symmetric polynomials, their inner products and the specialisation map.
These methods also have the advantage of being
applicable in far greater generality, as is evidenced by the fact that they were developed in \cite{TsuExt13} while classifying the irreducible
representations of certain logarithmic extensions of the minimal models.
This formalism also generalises to the more involved case of admissible level $\AKMA{sl}{2}$ theories \cite{DSSL2}.

We equate the minimal models with the simple \voas{} 
obtained by taking the quotient of the universal Virasoro \voas{} by their maximal
ideals at special values of the central charge.\footnote{Universal 
  means that we assume no relations on the
  defining field \(T(z)\) other than those required by the axioms of vertex
  operator algebras.}
The representation
theory of the minimal models can thus be obtained from that of the universal Virasoro \voas. Minimal model representations are just the
universal Virasoro \voa{} representations that are annihilated by the maximal ideal.
This elegant approach to classifying the representation theory of the minimal
models seems to have first been considered by Feigin, Nakanishi and
Ooguri \cite{FeiNul91} who applied it to a subset of the
minimal models, because, in general, having full computational control over the
maximal ideal is a very hard problem. However, free field realisations and
symmetric polynomials are exactly the tools one needs to solve this problem
and Theorem \ref{sec:specthm} extends the methods of annihilating ideals
\cite{FeiNul91} to all minimal models.

This note is organised as follows. Section \ref{sec:free-boson} gives an
overview of the free boson, the simplest example of a free field theory, as
well as vertex operators and screening operators. Section
\ref{sec:symmpoly} introduces symmetric polynomials and, in particular, gives an
overview of a one-parameter family of bases called the Jack symmetric polynomials. The
properties of these Jack polynomials are what yield such explicit
computational control that the representation theory of the minimal model representations
can be classified. Section \ref{sec:symmpoly} ends with explicit formulae for
singular vectors in terms of Jack polynomials. These formulae are originally
due to Mimachi and Yamada \cite{MimSing95}, though we give the new, much
simpler proof of \cite{TsuExt13}.
In Section \ref{sec:minmodspec}, the material of Sections \ref{sec:free-boson}
and \ref{sec:symmpoly} is combined to classify the highest weights of the irreducible
minimal model representations, in a similar manner to the methods of Feigin, Nakanishi and
Ooguri \cite{FeiNul91}.  This is then used to prove the complete reducibility of the representation theory without recourse to the (perhaps less familiar) methods of Zhu \cite{ZhuAlg96}.

\subsection*{Acknowledgements}

The second author would like to thank Akihiro Tsuchiya for introducing him to
the fascinating topics of free field realisations and symmetric polynomials.
The second author would also like to thank James Lepowsky and Siddhartha Sahi
for interesting discussions and Johannes Schmude for his extraordinary efforts in
literature searching.
Both authors thank Pierre Mathieu for interesting discussions
and his strong encouragement to write this note. 
The first author's research is supported by the Australian Research Council
Discovery Project DP1093910.
The second author's research is supported by the
Australian Research Council Discovery Early Career Researcher Award DE140101825.

\section{The free boson}
\label{sec:free-boson}

The free boson chiral algebra or Heisenberg vertex algebra is generated by a single
field \(a(z)\) which satisfies the operator product expansion
\begin{align}
  a(z)a(w)\sim\frac{1}{(z-w)^2}.
\end{align}
The Fourier expansion of the field \(a(z)\) is
\begin{align}
  a(z)=\sum_{n\in\mathbb{Z}}a_n z^{-n-1},
\end{align}
thus the operator product expansion implies the following commutations relations:
\begin{align}
  [a_m,a_n]=m\delta_{m,-n}\wun.
\end{align}
The Heisenberg Lie algebra \(\alg{H}\) is the infinite dimensional Lie algebra
generated by the \(a_n\) and the central element \(\wun\). We identify the
element \(\wun\) with the unit of the universal enveloping algebra \(U(\alg{H})\) of
\(\alg{H}\) and assume that \(\wun\) acts as the identity on any
\(\alg{H}\) 
representation.\footnote{This is only a minor restriction, since a simple rescaling of the generators \(a_n\)
  allows one to have the central element act as multiplication by any
  non-zero number.} The Heisenberg Lie algebra admits a triangular decomposition
\begin{equation}
  \begin{split}
    \alg{H}&=\alg{H}_-\oplus\alg{H}_0\oplus\alg{H}_+,\\
    \alg{H}_0&=\mathbb{C}a_0\oplus\mathbb{C}\wun,
  \end{split}
  \qquad
  \begin{split}
    \alg{H}_\pm&=\bigoplus_{n\geq1}\mathbb{C} a_{\pm n},\\
    \alg{H}_{\geq}&=\alg{H}_0\oplus \alg{H}_+.
  \end{split}
\end{equation}
The Verma modules \(\FF{\lambda}\), \(\lambda\in\mathbb{C}\),
with respect to this decomposition are called Fock modules. They are generated
by a highest weight vector \(\ket{\lambda}\) on which \(\alg{H}_{\geq}\) acts by
\begin{align}
  a_n\ket{\lambda} = \lambda \delta_{n,0} \ket{\lambda},\quad n\geq 0.
\end{align}
[Throughout this note, ``kets'' \(\ket{\lambda}\) will be reserved for the 
highest weight vectors of Fock modules and \(\lambda\) will denote the
highest weight.]
The \(\FF{\lambda}\) are then induced from \(\ket{\lambda}\) by
\begin{align}
  \FF{\lambda}=U(\alg{H})\otimes_{U(\alg{H}_{\geq})}\mathbb{C}\ket{\lambda}.
\end{align}
The parameter \(\lambda\) is called the \emph{Heisenberg weight}.
As is well known, the \(\FF{\lambda}\) are all irreducible.
As a vector space, 
\begin{align}\label{eq:creatoral}
  \FF{\lambda}\cong U(\alg{H}_{-})=\mathbb{C}[a_{-1},a_{-2},\dots].
\end{align}
As a representation over itself, the Heisenberg vertex algebra is identified with \(\FF{0}\) and
the operator state correspondence is given by
\begin{equation}
  \ket{0} \longleftrightarrow \wun,\qquad a_{-1}\ket{0} \longleftrightarrow
  a(z), \qquad 
  a_{-n_1-1}\cdots a_{-n_i-1}\ket{0}\longleftrightarrow
  \normord{\dfrac{\partial^{n_1}}{n_1!}a(z)\cdots \dfrac{\partial^{n_i}}{n_i!}a(z)}.
\end{equation}

The Heisenberg vertex algebra can be endowed with the structure of a vertex operator
algebra by choosing an energy-momentum tensor. This choice is not unique;
there is a one parameter family of choices:
\begin{align}\label{eq:T}
  T(z)=\frac{1}{2}:\alpha(z)^2:+\frac{\alpha_0}{2}\partial a(z),\quad \alpha_0\in\mathbb{C}.
\end{align}
The parameter \(\alpha_0\) determines the central charge of the
energy momentum tensor:
\begin{align}
  c=1-3\alpha_0^2.
\end{align}
The coefficients of the Fourier expansion of the energy momentum tensor are, by
definition, the generators \(L_n\) of the Virasoro algebra. Formula \eqref{eq:T}
identifies the Virasoro generators with infinite sums of elements of the
universal enveloping algebra \(U(\alg{H})\) of the Heisenberg Lie algebra:
\begin{align}
  \begin{split}
    T(z)&=\sum_{n\in\mathbb{Z}} L_n z^{-n-2}=\frac{1}{2}\sum_{n,m\in
      \mathbb{Z}}:a_ma_{n-m}:z^{-n-2}
    -\sum_{n\in\mathbb{Z}}\frac{\alpha_0}{2}(n+1)a_n z^{-n-2},\\
    L_n&=\frac{1}{2}\sum_{m\in\mathbb{Z}} :a_ma_{n-m}:
    -\frac{\alpha_0}{2}(n+1)a_n.
  \end{split}
\end{align}
This identification gives an action of the Virasoro algebra on the Fock
modules \(\FF{\lambda}\). The Fock modules thus become Virasoro highest weight
representations, that is,
\begin{equation}\label{eq:heistoconfweight}
    L_n\ket{\lambda} =h_\lambda \delta_{n,0}\ket{\lambda},\quad
    n\geq0, \qquad
    h_\lambda =\frac{1}{2}\lambda(\lambda-\alpha_0).
\end{equation}
Though the Fock modules are irreducible as Heisenberg representations, they need not
be so as Virasoro representations. 

The Heisenberg weights \(\lambda\) for which the conformal
weight \(h_\lambda\) is 1 play a special role as we shall see below. These weights are
roots of the degree 2 polynomial \(h_\lambda-1\) and we denote them by
\(\alpha_+,\alpha_-\). They satisfy the relations
\begin{align}
  \alpha_\pm=\frac{\alpha_0\pm\sqrt{\alpha_0^2+8}}{2},\qquad
  \alpha_++\alpha_-=\alpha_0,\qquad \alpha_+\alpha_-=-2.
\end{align}
\begin{thm}[Feigin-Fuchs \cite{FFFock90}]
  Let
  \begin{align}\label{eq:heislattice}
    \alpha_{r,s}=\frac{1-r}{2}\alpha_++\frac{1-s}{2}\alpha_-,\quad r,s\in\mathbb{Z}.
  \end{align}
  \begin{enumerate}
  \item For \(\alpha_+^2\in\mathbb{C}^\ast\) (or equivalently for
    \(\alpha_-^2\in\mathbb{C}^\ast\)), the Fock module \(\FF{\lambda}\) is reducible as a
    Virasoro representation if \(\lambda=\alpha_{r,s}\) for some
    \(r,s\in\mathbb{Z}\), \(rs>0\).
  \item If \(\alpha_+^2\) is non-rational (or equivalently if \(\alpha_-^2\)
    is non-rational), then the Fock module \(\FF{\lambda}\) is reducible as a
    Virasoro representation if and only if \(\lambda=\alpha_{r,s}\) for some
    \(r,s\in\mathbb{Z}\), \(rs>0\).
  \item If \(\alpha_+^2\) is positive rational (or equivalently if
    \(\alpha_-^2\) is positive
    rational), then the Fock module \(\FF{\lambda}\) is reducible as a
    Virasoro representation if and only if \(\lambda=\alpha_{r,s}\) for some \(r,s\in\mathbb{Z}\).
  \end{enumerate}
\end{thm}
\noindent We omit the corresponding result for negative rational \(\alpha_\pm^2\) as the
application to the minimal models does not require it.
We remark that Feigin and Fuchs also determined the precise structure of Fock modules as Virasoro representations in \cite{FFFock90}. For a comprehensive account of Virasoro representation theory, we
recommend the book by Iohara and Koga \cite{IohRep11}.

The work of Feigin and Fuchs shows that one can realise the universal Virasoro
vertex operator algebra at arbitrary central charge \(c\)
as a vertex operator subalgebra of the Heisenberg vertex
operator algebra. This free field realisation is called the \emph{Coulomb gas}
in the physics literature.

The Fock modules \(\FF{\lambda}\) with \(\lambda\neq0\) can be given a ``generalised'' vertex
algebra structure, that is, an operator state correspondence can also be
defined for the states of \(\FF{\lambda}\), though the operator product
expansions of these fields are generally not local. The operators
corresponding to the generating states \(\ket{\lambda}\in\FF{\lambda}\) are
called vertex operators in the physics literature.
These should not be confused with the fields (called chiral fields) of the \voa{}.

Before we can define vertex operators, we need to introduce an
operator \(\hat a\) whose commutation relations with the Heisenberg algebra are
\begin{align}
  [a_n,\hat a]=\delta_{n,0}\wun,\qquad [\wun,\hat a]=0.
\end{align}
The exponential of \(\hat a\) shifts weights, that is, for \(\lambda,\mu\in\mathbb{C}\),
\begin{align}
  a_0 e^{\mu\hat a}\ket{\lambda}=\mu e^{\mu\hat a}\ket{\lambda}
  +e^{\mu\hat a}a_0\ket{\lambda}=(\mu+\lambda)e^{\mu\hat a}\ket{\lambda}.
\end{align}
We identify \(e^{\mu\hat a}\ket{\lambda} =\ket{\mu+\lambda}\). Note that
\(e^{\mu\hat a}\) does not define a homomorphism of \(\alg{H}\) representations, since
it does not commute with \(\alg{H}\).

The vertex operator \(V_\lambda(z)\) corresponding to the state
\(\ket{\lambda}\) is
\begin{align}
  V_{\lambda}(z)=e^{\lambda\hat a}z^{\lambda a_0}
  \prod_{m\geq1}\exp\left(\lambda \frac{a_{-m}}{m}z^m\right)\prod_{m\geq1}\exp\left(-\lambda\frac{a_m}{m}z^{-m}\right).
\end{align}
The vertex operators are therefore linear maps
\begin{align}
  V_\lambda(z):\FF{\mu}\rightarrow \FF{\mu+\lambda}[[z,z^{-1}]]z^{\lambda\mu}.
\end{align}
The \(V_\lambda(z)\) are often defined as the
``normally ordered exponentials'' of a field
\begin{equation}
     \phi(z)=\hat{a} +a_o \log z -\sum_{n\neq0}\frac{a_n}{n}z^{-n},\quad
     V_\lambda(z)=:e^{\lambda \phi(z)}:.
 \end{equation}
Clearly \(\partial \phi(z)=a(z)\), which in turn implies that the operator product
expansions of \(\phi\) with itself and with \(a\) are
\begin{equation}
    a(z)\phi(w) \sim \frac{1}{z-w},\qquad
    \phi(z)\phi(w) \sim \log(z-w).
\end{equation}

Using these operator product expansions, one can verify that the
vertex operators \(V_\lambda\) are conformal primaries of conformal weight
\(h_\lambda\), that is, that
\begin{align}
  T(z)V_\lambda(w)\sim \frac{h_\lambda}{(z-w)^2}V_\lambda(w)+\frac{1}{z-w}\partial V_\lambda(w).
\end{align}
In particular, for \(h_{\alpha_{\pm}}=1\),
\begin{align}\label{eq:screenope}
  T(z)V_{\alpha_{\pm}}(w)\sim \frac{1}{(z-w)^2}V_{\alpha_\pm}(w)+\frac{1}{z-w}\partial V_{\alpha_{\pm}}(w)
  = \partial_w \frac{V_{\alpha_{\pm}}(w)}{z-w},
\end{align}
that is, the singular terms of these operator product expansions constitute total derivatives. 
Vertex operators with conformal weight 1 are called
\emph{screening operators} and were introduced by Dotsenko and Fateev \cite{DotScreen84}.
The conformal weight being 1 implies that the residue
\begin{align}
  \scr{\pm}=\frac{1}{2\pi\ii}\oint V_{\alpha_\pm}(w)\dd w
\end{align}
is a Virasoro homomorphism, because of
(\ref{eq:screenope}). In other words,
\begin{align}
  [T(z),\scr{\pm}]=\frac{1}{2\pi\ii}\oint T(z) V_{\alpha_\pm}(w) \dd w=0.
\end{align}
This residue is, of course, only well defined when the exponent of \(z^{\alpha_\pm a_0}\) is an integer.
Thus, for \(\alpha_\pm\mu\in\mathbb{Z}\), the residue of the vertex
operator \(V_{\alpha_\pm}(z)\) defines a Virasoro homomorphism
\begin{align}
  \scr{\pm}:\FF{\mu}\rightarrow \FF{\mu+\alpha_{\pm}}.
\end{align}

A natural question that one can ask in this context is whether these residues
can be generalised to obtain more Virasoro homomorphisms. The answer is ``yes'',
at least for suitable Heisenberg weights \(\mu\). The solution to generalising these Virasoro homomorphisms
lies in composing screening operators. The composition of \(n\) vertex
operators \(V_{\mu_i}\), \(i=1,\dots,n\), is given by
\begin{multline}\label{eq:screenprod}
  V_{\mu_1}(z_1)\cdots V_{\mu_n}(z_n)\\
  \ =e^{\hat a\sum_{i=1}^n\mu_i}\prod_{1\leq i<j\leq n}(z_i-z_j)^{\mu_i\mu_j}
  \prod_{i=1}^n z_i^{\mu_i a_0}
  \prod_{m\geq 1}\exp\left(\frac{a_{-m}}{m}\sum_{i=1}^n \mu_iz_i^m\right)
  \prod_{m\geq 1}\exp\left(-\frac{a_{m}}{m}\sum_{i=1}^n \mu_iz_i^{-m}\right).
\end{multline}
This formula is derived by using the operator product expansions above or by
using the commutation relations of the Heisenberg Lie algebra.
If we set \(\mu_i=\alpha_\pm\), \(i=1,\dots,n\), then the above formula
simplifies to
\begin{multline}
  V_{\alpha_\pm}(z_1)\cdots V_{\alpha_\pm}(z_n)\\
  \ =e^{n\alpha_\pm\hat a}\prod_{1\leq i<j\leq n}(z_i-z_j)^{\alpha_\pm^2}
  \prod_{i=1}^n z_i^{\alpha_\pm a_0}
  \prod_{m\geq 1}\exp\left(\alpha_\pm\frac{a_{-m}}{m}\sum_{i=1}^n z_i^m\right)
  \prod_{m\geq 1}\exp\left(-\alpha_\pm\frac{a_{m}}{m}\sum_{i=1}^n z_i^{-m}\right).
\end{multline}
We take the opportunity to introduce a family of symmetric polynomials, called
power sums, to simplify notation:
\begin{align}
  \fpowsum{m}{z}&=\sum_{i=1}^n z_i^m,\qquad \overline{\fpowsum{m}{z}}=\sum_{i=1}^n z_i^{-m}.
\end{align}
Up to a phase factor, which we suppress, the second factor of
\eqref{eq:screenprod} can be rewritten as
\begin{align}\label{eq:kappadef}
  \prod_{1\leq i\neq j\leq n}(z_i-z_j)^{\kappa_\pm},\quad \kappa_\pm=\frac{\alpha_\pm^2}{2}.
\end{align}
If we evaluate the product of these \(n\) screening operators on a Fock module
\(\FF{\mu}\), then the \(a_0\) generator acts by multiplication with \(\mu\)
and therefore
\begin{multline}
  V_{\alpha_\pm}(z_1)\cdots V_{\alpha_\pm}(z_n)\Big|_{\FF{\mu}}\\
  \ =e^{n\alpha_\pm\hat
    a}\prod_{1\leq i\neq j\leq n}(z_i-z_j)^{\kappa_\pm}
  \prod_{i=1}^n z_i^{\alpha_\pm \mu}
  \prod_{m\geq 1}\exp\left(\alpha_\pm\frac{a_{-m}}{m}\fpowsum{m}{z}\right)
  \prod_{m\geq 1}\exp\left(-\alpha_\pm\frac{a_{m}}{m}\overline{\fpowsum{m}{z}}\right).
\end{multline}

Let
\begin{equation}\label{eq:cyclenorm}
    c_n(\kappa_\pm) =\frac{2\pi \ii}{(n-1)!}
    \prod_{j=1}^{n-1}\frac{\Gamma(1+(j+1)\kappa_\pm)\Gamma(-j\kappa_\pm)}
    {\Gamma(\kappa_\pm+1)}.
\end{equation}
\begin{thm}[Tsuchiya-Kanie \cite{TsuScreen86}]\label{sec:tkthm}
  If \(d(d+1)\kappa_+\notin\mathbb{Z}\) and
    \(d(n-d)\kappa_+\notin\mathbb{Z}\), 
  for all integers \(d\) satisfying \(1\leq
  d\leq n-1\), then for each Heisenberg weight 
  \(\alpha_{n,k},\ k\in\mathbb{Z}\),
  there exists a cycle \(\Delta_n\) such that
  \begin{align} \label{eq:Contour+}
    \scrs{+}{n}=\frac{1}{c_n(\kappa_+)}\int_{\Delta_n} V_{\alpha_+}(z_1)\cdots V_{\alpha_+}(z_n)
    \dd z_1\cdots \dd z_n
  \end{align}
  is a non-trivial Virasoro homomorphism
  \begin{align}
    \scrs{+}{n}:\FF{\alpha_{n,k}}\rightarrow \FF{\alpha_{-n,k}}.
  \end{align}
  Likewise, if \(d(d+1)\kappa_-\notin\mathbb{Z}\) and
    \(d(n-d)\kappa_-\notin\mathbb{Z}\), for all integers \(d\) satisfying \(1\leq
  d\leq n-1\), then for each Heisenberg weight
  \(\alpha_{k,n}, k\in\mathbb{Z}\),
  there exists a cycle \(\Delta_n\) such that
  \begin{align} \label{eq:Contour-}
    \scrs{-}{n}=\frac{1}{c_n(\kappa_-)}\int_{\Delta_n} V_{\alpha_-}(z_1)\cdots V_{\alpha_-}(z_n)
    \dd z_1\cdots \dd z_n
  \end{align}
  is a non-trivial Virasoro homomorphism
  \begin{align}
    \scrs{-}{n}:\FF{\alpha_{k,n}}\rightarrow \FF{\alpha_{k,-n}}.
  \end{align}
  In particular,
  \begin{enumerate}
  \item if \(k\geq1\), then there exist vectors
    \(v\in\FF{\alpha_{n,k}}\) and \(w\in\FF{\alpha_{k,n}}\) such that
    \begin{equation}
      \scrs{+}{n}v=\ket{\alpha_{-n,k}}\in\FF{\alpha_{-n,k}},\qquad
      \scrs{-}{n}w=\ket{\alpha_{k,-n}}\in\FF{\alpha_{k,-n}},
    \end{equation}
    while \(\ket{\alpha_{n,k}}\) and \(\ket{\alpha_{k,n}}\) are annihilated by
    \(\scrs{+}{n}\) and \(\scrs{-}{n}\), respectively,
  \item if \(k\leq 0\), then
    \begin{equation}
      \scrs{+}{n}\ket{\alpha_{n,k}}\neq0\in \FF{\alpha_{-n,k}},\qquad
      \scrs{-}{n}\ket{\alpha_{k,n}}\neq0\in \FF{\alpha_{k,-n}}.
  \end{equation}
  \end{enumerate}
\end{thm}
\noindent 
The explicit construction of the cycles \(\Delta_n\) is rather subtle and we refer to
\cite{TsuScreen86} for the details. Intuitively, \(\Delta_n\) can be thought of as
\(n\) concentric circles about \(0\) that are pinched together at \(1\).

Let us try and understand the implications of Theorem \ref{sec:tkthm}
a little better.
Since \(\ket{\alpha_{n,k}}\) is a
Virasoro highest weight vector, then so is
\(\scrs{+}{n}\ket{\alpha_{n,k}}\) by virtue of
\(\scrs{+}{n}\) being a Virasoro homomorphism. Thus, whenever \(k\leq0\), the vectors
\(\scrs{+}{n}\ket{\alpha_{n,k}}\) and \(\scrs{-}{n}\ket{\alpha_{k,n}}\) 
generate Virasoro subrepresentations in
\(\FF{\alpha_{-n,k}}\) and \(\FF{\alpha_{k,-n}}\). Such Virasoro highest weight
vectors are called singular vectors.

For \(\mu_+=\alpha_{n,k},\ \mu_-=\alpha_{k,n}\),
\begin{align}
  \prod_{1\leq i\neq j\leq n}\left(z_i-z_j\right)^{\kappa_\pm}\prod_{i=1}^{n}z_i^{\alpha_\pm\mu_\pm}
  =\prod_{1\leq i\neq j\leq n}\left(1-\frac{z_i}{z_j}\right)^{\kappa_\pm}
  \prod_{i=1}^n z_i^{k-1},
\end{align}
where we have used the defining formulae \eqref{eq:heislattice} for 
\(\alpha_{n,k}, \alpha_{k,n}\) and \eqref{eq:kappadef} for \(\kappa_\pm\).
For later use we define
\begin{align}
  G_n(z;\kappa_\pm^{-1})=\prod_{1\leq i\neq j\leq n}\left(1-\frac{z_i}{z_j}\right)^{\kappa_\pm}.
\end{align}
This seemingly odd choice of \(\kappa_\pm^{-1}\) in the definition of \(G_n\) is
to make our notation in Section \ref{sec:symmpoly} conform with the
standard conventions in the symmetric polynomials literature \cite{MacSym95}.
The constant \(c_n(\kappa_\pm)\) in (\ref{eq:cyclenorm}) normalises the cycles \(\Delta_n\) 
in \eqref{eq:Contour+} and \eqref{eq:Contour-}
such that
\begin{equation}
  \frac{1}{c_n(\kappa_\pm)}\int_{\Delta_n}
  G_n(z:\kappa_\pm^{-1})
  \frac{\dd z_1\cdots \dd z_j}{z_1\cdots z_j}=1.
\end{equation}
Henceforth, we will therefore denote by \([\Delta_n]\) the homology class of the cycle
\(\Delta_n\) that has been rescaled by \(c_n(\kappa_\pm)^{-1}\), the
\(\kappa_\pm\)-dependence being left implicit,
that is,
\begin{align}
  \int_{[\Delta_n]}=\frac{1}{c_n(\kappa_\pm)}\int_{\Delta_n}.
\end{align}
The conditions on \(\kappa_\pm\) at the beginning of Theorem \ref{sec:tkthm}
ensure that \(c_n(\kappa_\pm)\neq 0\).
These conditions are met for the applications to the minimal models
in this note.
For a systematic discussion of how to
regularise \([\Delta_n]\) when these conditions are not met, see
\cite[Sections 3.2--3.4]{TsuExt13}.

By expanding the formulae for the Virasoro homomorphisms \(\scrs{\pm}{n}\) on \(\FF{\alpha_{n,k}}\) and \(\FF{\alpha_{k,n}}\), one sees that
\begin{align}\label{eq:virhomformula}
  \scrs{\pm}{n}=e^{n\alpha_\pm\hat a}\int_{[\Delta_n]}
  G_n(z;\kappa_\pm^{-1})\prod_{i=1}^n z_i^{k}
   \prod_{m\geq 1}\exp\left(\alpha_\pm\frac{a_{-m}}{m}\fpowsum{m}{z}\right)
  \prod_{m\geq 1}\exp\left(-\alpha_\pm\frac{a_{m}}{m}\overline{\fpowsum{m}{z}}\right)
  \frac{\dd z_1\cdots\dd z_n}{z_1\cdots z_n}.
\end{align}
Apart from the multivalued function \(G_n\), the integrand consists of an
infinite sum of monomials in \(U(\mathfrak{H}_-)\otimes U(\mathfrak{H}_+)\)
where the coefficients are
products of polynomials in either positive or negative powers of the \(z_i\).
It turns out that for symmetric polynomials \(f,g\), the pairing
\begin{align}\label{eq:inprod}
  \jprod{f,g}{\kappa_\pm^{-1}}{n}=
  \int_{[\Delta_n]}
  G_n(z;\kappa_\pm^{-1}) f(z_1,\dots,z_n)
  \overline{g(z_1,\dots,z_n)}\frac{\dd z_1\cdots \dd z_n}{z_1\cdots z_n},
\end{align}
where \(\overline{g(z_1,\dots,z_n)}=g(z_1^{-1},\dots,z_n^{-1})\), defines a
non-degenerate symmetric bilinear form on the ring of symmetric
polynomials in \(n\) variables. Evaluating the Virasoro homomorphism
\(\scrs{\pm}{n}\) in (\ref{eq:virhomformula})
therefore reduces to evaluating inner products of symmetric polynomials.
As we shall see, this is a well known problem with an elegant solution.

\section{Symmetric polynomials}
\label{sec:symmpoly}

For a comprehensive study of symmetric polynomials, we recommend the book by
Macdonald \cite{MacSym95}.
Let \(\Lambda_n\) be the ring of symmetric polynomials with complex
coefficients. As a commutative ring, \(\Lambda_n\) is generated
by a number of interesting sets of polynomials 
including the elementary symmetric polynomials
\begin{align}
  \felsym{i}{z}=\sum_{1\leq j_1<\cdots <j_i\leq n}z_{j_1}\cdots z_{j_i},\quad i=1,\dots,n
\end{align}
and the power sums
\begin{align}
  \fpowsum{i}{z}=\sum_{j=1}^n z_j^i,\quad i=1,\dots,n.
\end{align}
These polynomials are algebraically independent and generate \(\Lambda_n\)
freely, that is,
\begin{align}
  \Lambda_n=\mathbb{C}[\elsym{1},\dots,\elsym{n}]=\mathbb{C}[\powsum{1},\dots,\powsum{n}].
\end{align}
The ring \(\Lambda_n\) is clearly also a complex vector space and it is
natural to look for convenient bases. One such basis is constructed from the
power sums. Let \(\lambda=(\lambda_1,\dots,\lambda_k)\), \(k\geq 0\), be a
partition of an integer with largest part \(\lambda_1\leq n\) (we follow the
convention of listing the parts in weakly descending order). Then, for all such \(\lambda\), the
\begin{align}\label{eq:powsumprods}
  \fpowsum{\lambda}{z}=\fpowsum{\lambda_1}{z}\cdots \fpowsum{\lambda_k}{z}
\end{align}
are linearly independent and form a basis of \(\Lambda_n\). 
Another convenient basis of \(\Lambda_n\) is given by the monomial symmetric polynomials.
Let \(\mu=(\mu_1,\dots,\mu_n)\) be a partition of length \(\ell(\mu)\) at most \(n\) (if the
partition is shorter than \(n\) pad it with 0s at the end until it is length \(n\)).
Then, the monomial symmetric polynomials
are defined as
\begin{align}
  \fmonsym{\mu}{z}=\sum_{\tau} z_1^{\tau_1}\cdots z_n^{\tau_n},\quad
  \tau\in\{\text{all distinct permutations of }\mu\}.
\end{align}
We shall refer to these polynomials as the symmetric monomials for brevity.

As we can see from the power sums and the symmetric monomials, the set of
partitions that label basis elements must be truncated once the weight 
\(|\lambda|=\sum_i\lambda_i\) is greater than the number of
variables.
Specifically, there exist partitions with \(\lambda_1>n\), which are not allowed
for the power sums, or \(\ell(\lambda)>n\), which are not allowed for the symmetric monomials.
This is why
it is convenient to work in the limit of infinitely many variables:
\begin{align}
  \Lambda=\varprojlim_{n}\Lambda_n.
\end{align}
One can then easily recover the finite variable case by the projection
\begin{align}
    \gamma_n:\Lambda&\rightarrow \Lambda_n\\\nonumber
    z_j&\mapsto
      \begin{cases}
        z_j&1\leq j\leq n\\
        0&j>n
      \end{cases}
\end{align}
that sets to 0 all but the first \(n\) variables.
The power sums in infinitely many variables now generate \(\Lambda\) as 
their finite-variable versions did \(\Lambda_n\).  We continue to use \eqref{eq:powsumprods} 
to define \(\powsum{\lambda}\) in the infinite-variable case.
\begin{align}\label{eq:sympolyalg}
  \Lambda=\mathbb{C}[\powsum1,\powsum2,\powsum3,\dots].
\end{align}
The power sum and symmetric monomial bases of \(\Lambda\) are now labelled
by all partitions of integers without restrictions on parts or lengths:
\begin{align}
  \Lambda=\bigoplus_{\lambda}\mathbb{C} \powsum\lambda=\bigoplus_{\lambda}\mathbb{C}\monsym\lambda.
\end{align}
The projection to the finite variable case is particularly easy in the
symmetric monomial basis:
\begin{align}
    \gamma_n:\Lambda&\rightarrow \Lambda_n\\\nonumber
    \monsym\mu&\mapsto
      \begin{cases}
        \fmonsym{\mu}{z_1,\dots,z_n}&\ell(\mu)\leq n\\
        0&\text{else}
      \end{cases}.
\end{align}
Recall that, by \eqref{eq:creatoral}, the universal
enveloping algebra \(U(\alg{H}_-)\) is also a polynomial algebra in infinitely many
generators. Identifying these two algebras will be important below.

\begin{prop}\label{sec:innerprod}
  For \(f,g\in\Lambda_n\) and \(\kappa\in\mathbb{C}^\ast\) such that
  \(d/\kappa\notin\mathbb{Z}\) for all integers satisfying \(1 \leq d\leq
  n\), the bilinear form
  \begin{align}
    \jprod{f,g}{\kappa}{n}=
    \int_{[\Delta_n]}
    G_n(z;\kappa) f(z_1,\dots,z_n)
    \overline{g(z_1,\dots,z_n)}\frac{\dd z_1\cdots \dd z_n}{z_1\cdots z_n}
  \end{align}
  is
  \begin{enumerate}
  \item symmetric,
  \item non-degenerate,
  \item graded: \(\jprod{f,g}{\kappa}{n}=0\) if \(\deg f\neq \deg g\).
  \end{enumerate}
\end{prop}

Proposition \ref{sec:innerprod} leads us to the basis of \(\Lambda_n\) that is
most important for our purposes, the Jack polynomials \(\fjack{\lambda}{\kappa}{z}\). These polynomials
are characterised by two properties \cite{MacSym95}:
\begin{enumerate}
\item The Jack polynomials have upper triangular expansions in the basis of
  symmetric monomials with respect to the dominance ordering of
  partitions\footnote{The dominance ordering is a partial ordering of
    partitions of equal weight defined by
    \begin{displaymath}
      \lambda\geq \mu \iff \sum_{i=1}^k\lambda_i\geq
      \sum_{i=1}^k\mu_i,\quad \forall k\geq1.
    \end{displaymath}}, that is,
  \begin{align}
    \fjack{\lambda}{\kappa}{z}=\sum_{\lambda\geq\mu}u_{\lambda,\mu}(\kappa) \fmonsym{\mu}{z},
  \end{align}
  where the \(u_{\lambda,\mu}(\kappa)\in\mathbb{C}\) and \(u_{\lambda,\lambda}(\kappa)=1\).
\item The Jack polynomials are mutually orthogonal:
  \begin{align}
    \jprod{\fjack{\lambda}{\kappa}{z},\fjack{\mu}{\kappa}{z}}{\kappa}{n}=0,\quad
    \text{if }\lambda\neq\mu.
  \end{align}
\end{enumerate}
Since the dominance ordering of partitions is only a partial ordering, trying
to determine the Jack polynomials by means of Gram-Schmidt orthogonalisation is
an overdetermined problem. Showing that they exist 
is therefore non-trivial, see \cite{MacSym95}.

We prepare some notation regarding partitions. For a partition \(\lambda\), 
let \(s=(i,j)\in\lambda\) be a box in the Young tableau of \(\lambda\), 
so that \(i=1,\dots,\ell(\lambda)\) and \(j=1,\dots,\lambda_i\). Then, the arm length,
coarm length, leg length and coleg length are defined to be
\begin{equation}
  \begin{split}
    a(s)=\lambda_i-j,\quad a^\prime(s)=j-1, \quad
    l(s)=\lambda^\prime_j-i, \quad l^\prime(s)=i-1, 
  \end{split}
\end{equation}
respectively, where \(\lambda^\prime\) is the conjugate partition of \(\lambda\), that is,
the partition for which the columns and rows of the Young tableau have been exchanged.
\begin{prop}\label{sec:Jackprops}\
  \begin{enumerate}
  \item Jack polynomials exist (for all \(n\) and in the infinite variable limit).
  \item The Jack polynomials satisfy the same projection formulae as the
    symmetric monomials:
    \begin{equation}
      \begin{split}
        \gamma_n(\fjack{\lambda}{\kappa}{z})=
          \begin{cases}
            \fjack{\lambda}{\kappa}{z_1,\dots,z_n}&\ell(\lambda)\leq n\\
            0&\mathrm{else}
          \end{cases}.
        \end{split}
    \end{equation}
  \item For either a finite or infinite number of variables \(z_i,y_j\),
    \begin{equation}
      \begin{split}
        \prod_{i,j\geq 1}(1-z_i y_j)^{-1/\kappa}&
        =\prod_{m\geq1}\exp\left(\frac{1}{\kappa}\frac{\fpowsum{m}{z}\fpowsum{m}{y}}{m}\right)
        =\sum_{\lambda}b_{\lambda}(\kappa)\fjack{\lambda}{\kappa}{z}\fjack{\lambda}{\kappa}{y},\\
        b_\lambda(\kappa)&=\prod_{s\in\lambda}\frac{a(s)\kappa+l(s)+1}{(a(s)+1)\kappa+l(s)}.
      \end{split}
    \end{equation}
  \item The norm of the mutually orthogonal Jack polynomials is
    \begin{align}
      \jprod{\fjack{\lambda}{\kappa}{z},\fjack{\lambda}{\kappa}{z}}{\kappa}{n}=
      \prod_{s\in\lambda}\frac{(a(s)+1)\kappa+l(s)}{a(s)\kappa+l(s)+1}
      \frac{n+a^\prime(s)\kappa-l^\prime(s)}{n+(a^\prime(s)+1)\kappa-l^\prime(s)-1}.
    \end{align}
  \item For \(X\in\mathbb{C}\), let \(\Xi_X:\Lambda\rightarrow \mathbb{C}\) be
    the algebra homomorphism defined, in the power
    sum basis, by
    \begin{align}
      \Xi_X(\powsum\lambda(y))=X^{\ell(\lambda)}.
    \end{align}
    The map \(\Xi_X\) is called the \emph{specialisation map}.
    Then,
    \begin{align}
      \Xi_X\left(\fjack{\lambda}{\kappa}{y}\right)=\prod_{s\in\lambda}
      \frac{X+a^\prime(s)\kappa-l^\prime(s)}{a(s)\kappa+l(s)+1}
    \end{align}
    and
    \begin{align}
      \prod_{i\geq1}(1-z_i)^{-X/\kappa}=\prod_{m\geq1}\exp\left(\frac{X}{\kappa}\frac{\fpowsum{m}{z}}{m}\right)
      =\sum_\lambda b_\lambda(\kappa)\fjack{\lambda}{\kappa}{z} \Xi_X(\fjack{\lambda}{\kappa}{y}).
    \end{align}
    We stress that while this homomorphism applies to symmetric polynomials in any
    variables, we will only be applying it to those in the \(y\) variables.
  \item Let \((m^n)=(m,\dots,m)\) be the partition consisting of \(n\)
    copies of \(m\). Then,
    \begin{align}
      \gamma_n(\fjack{(m^n)}{\kappa}{z})=\fjack{(m^n)}{\kappa}{z_1,\dots,z_n}=\fmonsym{(m^n)}{z_1,\dots,z_n}=
      \prod_{i=1}^n z_i^m.
    \end{align}
  \end{enumerate}
\end{prop}
\noindent
See \cite{MacSym95} for proofs.

Armed with this knowledge of Jack polynomials, we can
now explicitly evaluate the action of screening operators on Fock
modules. Recall from equations \eqref{eq:creatoral} and \eqref{eq:sympolyalg} that both 
\(U(\alg{H}_-)\) and \(\Lambda\) are polynomial algebras in an infinite number
of variables,
\begin{align}
  \mathbb{C}[a_{-1},a_{-2},\dots]=U(\alg{H}_{-})\cong \Lambda=\mathbb{C}[\powsum1,\powsum2,\dots],
\end{align}
and are therefore isomorphic. For \(\delta\in\mathbb{C}\), we define the algebra isomorphism
\begin{align}
  \rho_\delta:\Lambda &
  \rightarrow U(\alg{H}_-),\\\nonumber
  \fpowsum{n}{y} &
  \mapsto \delta a_{-n}.
\end{align}
As with the specialisation map, we will only be applying \(\delta\) to polynomials in the \(y\) variables.
\begin{thm}\label{sec:singvecformula}
  For \(k\geq0\), let \((k^n)=(k,\dots,k)\) be the partition consisting of \(n\)
  copies of \(k\). Then, the Virasoro homomorphisms
  \begin{equation}
    \scrs{+}{n}:\FF{\alpha_{n,-k}}\rightarrow\FF{\alpha_{-n,-k}},\qquad
    \scrs{-}{n}:\FF{\alpha_{-k,n}}\rightarrow\FF{\alpha_{-k,-n}}
  \end{equation}
  map the vectors \(\ket{\alpha_{n,-k}}\) and \(\ket{\alpha_{-k,n}}\) to the non-zero singular vectors
  \begin{equation}
    \begin{split}
      \scrs{+}{n}\ket{\alpha_{n,-k}}&=
      b_{(k^n)}(\kappa_+^{-1})\rho_{\frac{2}{\alpha_+}}\left(\fjack{(k^n)}{\kappa_+^{-1}}{y}\right)
      \ket{\alpha_{-n,-k}},\\
      \scrs{-}{n}\ket{\alpha_{-k,n}}&=
      b_{(k^n)}(\kappa_-^{-1})\rho_{\frac{2}{\alpha_-}}\left(\fjack{(k^n)}{\kappa_-^{-1}}{y}\right)
      \ket{\alpha_{-k,-n}}.
    \end{split}
  \end{equation}
\end{thm}
\begin{proof}
  We prove the formula for \(\scr{+}\). The one for \(\scr{-}\) follows similarly.
  The proof follows by direct evaluation using the theory of Jack polynomials:
  \begin{align}
      \scrs{+}{n}\ket{\alpha_{n,-k}}
      &=\int_{[\Delta_n]}G_n(z;\kappa_+^{-1}) \prod_{i=1}^n
      z_i^{-k}\prod_{m\geq1}\exp\left(\alpha_+\frac{\fpowsum{m}{z}}{m}a_{-m}\right)
      \ket{\alpha_{-n,-k}}\frac{\dd z_1\cdots\dd z_n}{z_1\cdots z_n}\nonumber\\
      &\overset{1}{=}\jprod{\fjack{(k^n)}{\kappa_+^{-1}}{z},
      \prod_{m\geq1}\exp\left({\alpha_+\frac{\fpowsum{m}{z}}{m}a_{-m}}\right)}{\kappa_+^{-1}}{n}
      \ket{\alpha_{-n,-k}}\nonumber\\
      &=\jprod{\fjack{(k^n)}{\kappa_+^{-1}}{z},\rho_{\frac{2}{\alpha_+}}\left(\prod_{m\geq1}
        \exp\left(\kappa_+\frac{\fpowsum{m}{z}\fpowsum{m}{y}}{m}\right)\right)}{\kappa_+^{-1}}{n}
      \ket{\alpha_{-n,-k}}\nonumber\\
      &\overset{2}{=}\sum_\mu b_\mu(\kappa_+^{-1})\jprod{
      \fjack{(k^n)}{\kappa_+^{-1}}{z},\fjack{\mu}{\kappa_+^{-1}}{z}}{\kappa_+^{-1}}{n}
      \rho_{\frac{2}{\alpha_+}}\left(
        \fjack{\lambda}{\kappa_+^{-1}}{y}\right)
      \ket{\alpha_{-n,-k}}\nonumber\\
      &\overset{3}{=}b_{(k^n)}(\kappa_+^{-1})\jprod{
      \fjack{(k^n)}{\kappa_+^{-1}}{z},\fjack{(k^n)}{\kappa_+^{-1}}{z}}{\kappa_+^{-1}}{n}
    \rho_{\frac{2}{\alpha_+}}\left(
        \fjack{(k^n)}{\kappa_+^{-1}}{y}\right)
      \ket{\alpha_{-n,-k}}\nonumber\\
      &\overset{4}{=}b_{(k^n)}(\kappa_+^{-1})\rho_{\frac{2}{\alpha_+}}\left(
        \fjack{(k^n)}{\kappa_+^{-1}}{y}\right)
      \ket{\alpha_{-n,-k}}.
  \end{align}
  Here we have used item (6) of Proposition \ref{sec:Jackprops} for
  \(\overset{1}{=}\) to identify
  \begin{equation}
  \overline{\fjack{(k^n)}{\kappa_+^{-1}}{z}}=\prod_{i=1}^n z_{i}^{-k};
\end{equation}
  item (3) of Proposition \ref{sec:Jackprops} for \(\overset{2}{=}\)
  (remembering that the integration in the inner product is over the \(z\) variables);
  the orthogonality of Jack polynomials for \(\overset{3}{=}\);
  and item (6) of Proposition \ref{sec:Jackprops} to see that 
  \begin{equation}
    \fjack{(k^n)}{\kappa_+^{-1}}{z}\overline{\fjack{(k^n)}{\kappa_+^{-1}}{z}}=1\quad \Rightarrow\quad
    \jprod{\fjack{(k^n)}{\kappa_+^{-1}}{z},\fjack{(k^n)}{\kappa_+^{-1}}{z}}{\kappa_+^{-1}}{n}=1,
  \end{equation}
  which justifies \(\overset{4}{=}\). By direct evaluation of formula for
  \(b_{(k^n)}(\kappa_+^{-1})\) in item (3) of Proposition \ref{sec:Jackprops},
  one sees that \(b_{(k^n)}(\kappa_+^{-1})\) is a product of quotients of
  positive rational numbers and is therefore non-zero.
\end{proof}
\noindent
This theorem is originally due to Mimachi and Yamada \cite{MimSing95}, though the much simpler and
streamlined proof that we have presented here first appeared in \cite[Proposition 3.24]{TsuExt13}.

\section{The minimal models and their representations}
\label{sec:minmodspec}

In Section \ref{sec:free-boson}, we constructed the universal Virasoro vertex
operator algebra at central charge
\begin{align}
  c=1-3\alpha_0^2,\quad \alpha_0\in\mathbb{C},
\end{align}
as a vertex operator subalgebra of the Heisenberg vertex operator algebra. At
generic values of the central charge \(c\) (or equivalently, at generic values
of \(\alpha_0\)), the universal Virasoro vertex operator algebra is simple and
contains no non-trivial ideals. However, there is a discrete set of central
charges at which the universal Virasoro vertex operator algebra is not
simple. The minimal model vertex operator algebras are the simple vertex
operator algebras obtained, for these central charges, by taking the quotients
of the universal vertex operator algebras by their maximal ideals. Thus, the minimal model vertex
operator algebras can be realised as subquotients of Heisenberg vertex
operator algebras.

The minimal model central charges, that is, the central
charges at which the
universal Virasoro \voas{} are non-simple, are precisely
\begin{align}
  c_{p_+,p_-}=c_{p_-,p_+}=1-6\frac{(p_+-p_-)^2}{p_+p_-},
\end{align}
where \(p_+,p_-\geq2\) are coprime integers \cite{FFFock90}.
We denote the minimal model
vertex operator algebra of central charge \(c_{p_+,p_-}\) by \(\MinMod{p_+}{p_-}\).
To obtain these minimal model central charges for the Heisenberg algebra, we
set
\begin{align}
  \alpha_0=\alpha_++\alpha_-,\qquad\alpha_+=\sqrt{\frac{2p_-}{p_+}},\qquad\alpha_-=-\sqrt{\frac{2p_+}{p_-}}.
\end{align}
The parameters \(\alpha_\pm\) are precisely the Heisenberg weights which, by 
formula \eqref{eq:heistoconfweight}, correspond to conformal weight 1. 
The \(\kappa_\pm\) parameters introduced in Section \ref{sec:free-boson} are thus,
\begin{align}
  \kappa_+=\frac{\alpha_+^2}{2}=\frac{p_-}{p_+},\qquad\kappa_-=\frac{\alpha_-^2}{2}=\frac{p_+}{p_-}.
\end{align}

The ideal \(\MinIdeal{p_+}{p_-}\) of the universal Virasoro \voa{} of central charge \(c_{p_+,p_-}\) is
generated by a singular vector of 
conformal weight \((p_+-1)(p_--1)\) \cite{FFFock90}. By using the screening operator
formalism of Section \ref{sec:free-boson}, in particular Theorem
\ref{sec:singvecformula}, we can realise this singular vector using the 
screening operator \(\scr{+}(z)=V_{\alpha_+}(z)\) or \(\scr{-}(z)=V_{\alpha_-}(z)\). Writing
\begin{align}
  \scrs{\pm}{n}
  =\int_{[\Delta_n]}V_{\alpha_\pm}(z_1)\cdots V_{\alpha_\pm}(z_n)\dd
    z_1\cdots \dd z_n,
\end{align}
we deduce that the singular vector in \(\FF{0}\) which generates the ideal of the
universal Virasoro vertex operator algebra, sitting inside the Heisenberg
vertex operator algebra, is given by
\begin{equation}\label{eq:singvects}
  \begin{split}
    \scrs{+}{p_+-1}\ket{(1-p_+)\alpha_+}&=b_{((p_--1)^{p_+-1})}(\kappa_+^{-1})
    \rho_{\frac{2}{\alpha_+}}\left(\fjack{((p_--1)^{p_+-1})}{\kappa_+^{-1}}{y}\right)\ket{0},
    \\
    \scrs{-}{p_--1}\ket{(1-p_-)\alpha_-}&=b_{((p_+-1)^{p_--1})}(\kappa_-^{-1})
    \rho_{\frac{2}{\alpha_-}}\left(\fjack{((p_+-1)^{p_--1})}{\kappa_-^{-1}}{y}\right)\ket{0}.
  \end{split}
\end{equation}
The above equations are obtained directly from Theorem
\ref{sec:singvecformula}, the first by choosing \(n=p_+-1,\ k=1-p_-\) and the second by choosing
\(n=p_--1,\ k=1-p_+\).
For a given conformal weight, the Virasoro singular vectors of a Fock module
are unique up to rescaling \cite{AstStr97}, if they exist, so the two vectors in \eqref{eq:singvects} are
proportional to each other.

As a final demonstration of the power of combining the screening operator and
symmetric polynomial formalisms,
we will classify the representations of the minimal model
vertex operator algebras. Since the universal Virasoro vertex operator algebras
are subalgebras of the Heisenberg \voas{}, the Fock modules
\(\FF{\mu}\) are representations of the universal Virasoro vertex operator
algebras for any \(\mu\in\mathbb{C}\). However, the Virasoro representation generated
from \(\ket{\mu}\) can only be a representation
of \(\MinMod{p_+}{p_-}\) if each field corresponding to a vector in the ideal
\(\MinIdeal{p_+}{p_-}\) acts trivially. Moreover, any irreducible highest weight
representation of \(\MinMod{p_+}{p_-}\) must be realisable as a subquotient of
a Fock module as, for any conformal weight,
there exists a Fock module whose generating vector has that conformal weight, by \eqref{eq:heistoconfweight}.
\begin{thm}\label{sec:specthm}
  Let
  \begin{align}
    h_{r,s}=\frac{(rp_--sp_+)^2-(p_+-p_-)^2}{4p_+p_-}.
  \end{align}
  Up to isomorphism, there are exactly \(\frac{1}{2}(p_+-1)(p_--1)\)
  inequivalent irreducible \(\MinMod{p_+}{p_-}\) representations. 
  They are given by the irreducible representations of the Virasoro algebra
  generated by highest weight vectors of conformal weight
  \begin{align}
    h_{r,s},\quad 1\leq r\leq p_+-1,\ 1\leq s\leq p_--1,\ rp_-+sp_+\leq p_+p_-.
  \end{align}
\end{thm}
\begin{proof}
  We only prove that the above list of irreducible representations of the
  Virasoro algebra is an upper
  bound on the set of inequivalent irreducible \(\MinMod{p_+}{p_-}\) representations. In
  order to show that the list is saturated, one can then either construct all
  these representations by, for example, the coset construction
  \cite{GodCoset85,KacCosets88},
  by quantum hamiltonian reduction \cite{FeiDSred90},
  or use Zhu's algebra \cite{WangRat93,ZhuAlg96},
  this being the associative algebra of zero modes of the fields of the \voa{}
  acting on highest weight vectors.

  Consider the singular vector
  \begin{align}
    \ket{\chi}=\scrs{+}{p_+-1}\ket{(1-p_+)\alpha_+}\in\FF{0}
  \end{align}
  of equation (\ref{eq:singvects}).
  The corresponding field is obtained by
  integrating \(p_+-1\) vertex operators \(V_{\alpha_+}\) over \([\Delta_{p_+-1}]\) about the vertex
  operator \(V_{(1-p_+)\alpha_+}\), with \([\Delta_{p_+-1}]\) centred about
  the argument of \(V_{(1-p_+)\alpha_+}\):
  \begin{align}
    \chi(w) =\int_{[\Delta_{p_+-1}]}\scr{+}(z_1+w)\cdots\scr{+}(z_{p_+-1}+w)
      V_{(1-p_+)\alpha_+}(w)\dd z_1\cdots \dd z_{p_+-1}.
  \end{align}
  The vector \(\ket{\chi}\) is an element of the ideal \(\MinIdeal{p_+}{p_-}\), so
  the field \(\chi(w)\) must therefore act trivially on any
  \(\MinMod{p_+}{p_-}\) representation.
  Consequently,
  \begin{align}
    \bracket{\mu}{\chi(w)}{\mu}=0,
  \end{align}
  where \(\ket{\mu}\) is the highest weight vector of \(\FF{\mu}\) and
  \(\bra{\mu}\) its dual which satisfies
  \begin{equation}
    \braket{\mu}{\mu}=1,\qquad \bra{\mu}a_n=\bra{\mu}\delta_{n,0}\mu,\quad n\leq0.
  \end{equation}
  
  We can evaluate \(\bracket{\mu}{\chi(w)}{\mu}\) using the theory of
  Jack polynomials.
  Applying formula \eqref{eq:screenprod} to simplify the composition of vertex
  operators in the definition of \(\chi(w)\), we see that
  \begin{align}\label{eq:unintsingfield}
      &
      \scr{+}(z_1+w)\cdots
      \scr{+}(z_{p_+-1}+w)V_{(1-p_+)\alpha_+}(w)\nonumber\\
      &
      \mspace{150mu} =
      G_{p_+-1}(z;\kappa_+^{-1})
      \prod_{i=1}^{p_+-1} (z_i+w)^{\alpha_+ a_0}\cdot w^{(1-p_+)\alpha_+a_0}\nonumber\\
      &
      \mspace{150mu} \quad\times\prod_{m\geq 1}\exp\left(\alpha_+\frac{a_{-m}}{m}\left(\fpowsum{m}{z_1+w,\dots,z_{p_+-1}+w}+(1-p_+)w^m\right)\right)\nonumber\\
      &
      \mspace{150mu} \quad\times\prod_{m\geq 1}\exp\left(-\alpha_+\frac{a_{m}}{m}\left(\overline{\fpowsum{m}{z_1+w,\dots,z_{p_+-1}+w}}+(1-p_+)w^{-m}\right)\right).
  \end{align}
  The exponentials of Heisenberg generators \(a_m\), \(m\neq 0\), in
  \eqref{eq:unintsingfield} annihilate \(\bra{\mu}\) and \(\ket{\mu}\).
  Thus,
  \begin{align}
      \bracket{\mu}{\chi(w)}{\mu}&= \int_{[\Delta_{p_+-1}]}
      \bracket{\mu}{\scr{+}(z_1+w)\cdots \scr{+}(z_{p_+-1}+w) V_{(1-p_+)\alpha_+}(w)}{\mu}\dd
      z_1\cdots \dd z_{p_+-1}\nonumber\\
      &=\int_{[\Delta_{p_+-1}]}
      G_{p_+-1}(z;\kappa_+^{-1})\prod_{i=1}^{p_+-1}z_i^{1-p_-}\cdot\prod_{i=1}^{p_+-1}(z_i+w)^{\alpha_+\mu}\cdot
      w^{(1-p_+)\alpha_+\mu}\frac{\dd z_1\cdots z_{p_+-1}}{z_1\cdots
        z_{p_+-1}}\nonumber\\
      &=\int_{[\Delta_{p_+-1}]}
      G_{p_+-1}(z;\kappa_+^{-1})\prod_{i=1}^{p_+-1}z_i^{-(p_--1)}\cdot\prod_{i=1}^{p_+-1}\left(1+\frac{z_i}{w}\right)^{\alpha_+\mu}
      \frac{\dd z_1\cdots\dd z_{p_+-1}}{z_1\cdots z_{p_+-1}}\nonumber\\
      &=\jprod{
      \fjack{((p_--1)^{p_+-1})}{\kappa_+^{-1}}{z},\prod_{i=1}^{p_+-1}\left(1+\frac{z_i}{w}\right)^{\alpha_+\mu}
      }{\kappa_+^{-1}}{p_+-1}\nonumber\\
      &=(-w)^{-(p_+-1)(p_--1)}\ b_{((p_--1)^{p_+-1})}(\kappa_+^{-1})
      \ \Xi_{\alpha_-\mu}\left(\fjack{((p_--1)^{p_+-1})}{\kappa_+^{-1}}{y}\right)\nonumber\\
      &=(-w)^{-(p_+-1)(p_--1)}\prod_{s\in(p_--1)^{p_+-1}}
      \frac{\alpha_-\mu+a^\prime(s)/\kappa_+-l^\prime(s)}{(a(s)+1)/\kappa_++l(s)}\nonumber\\
      &=(-w)^{-(p_+-1)(p_--1)}\prod_{i=1}^{p_+-1}\prod_{j=1}^{p_--1}
      \frac{\alpha_-\mu+(j-1)/\kappa_++1-i}{(p_--j)/\kappa_++p_+-1-i},
  \end{align}
  where we have evaluated the inner product using item (5) of Proposition \ref{sec:Jackprops}.
  Clearly, the denominator of the above product is non-singular, since
  \(\kappa_+={p_-}/{p_+}\) is a positive rational number. Therefore,
  \(\bracket{\mu}{\chi(w)}{\mu}=0\) whenever
  \begin{align}
    \begin{split}
      0&=\prod_{i=1}^{p_+-1}\prod_{j=1}^{p_--1}\brac{\alpha_-\mu+(j-1)/\kappa_++1-i}
      =C \prod_{i=1}^{p_+-1}\prod_{j=1}^{p_--1}\brac{\mu-\alpha_{i,j}}, 
    \end{split}
  \end{align}
  where \(C\) is a non-zero constant.  We group the \((i,j)\)-factor with the \((p_+-i,p_--j)\)-factor:
  \begin{align}
    \left(\mu-\alpha_{i,j}\right)\cdot
    \left(\mu-\alpha_{p_+-i,p_--j}\right)
    =2 h_\mu - 2 h_{i,j}.
  \end{align}
  Thus,
  \begin{equation}
    \bracket{\mu}{\chi(w)}{\mu}=0\iff
    \prod_{(r,s)}(h_\mu-h_{r,s})=0,
  \end{equation}
  where the index \((r,s)\) runs over all \(1\leq r\leq p_+-1\) and
  \(1\leq s\leq p_--1\), with \(rp_-+sp_-<p_+p_-\).
  The above constraints imply that the conformal highest weight of an 
  \(\MinMod{p_+}{p_-}\) representation must be a root of the polynomial
  \begin{equation}\label{eq:minpoly}
    f(h)=\prod_{(r,s)}(h-h_{r,s}),
  \end{equation}
  that is, it must be equal to \(h_{r,s}\) for some \(1\leq r\leq p_+-1\) and
  \(1\leq s\leq p_--1\), with \(rp_-+sp_-<p_+p_-\). 
\end{proof}

Showing that the representation theory of \(\MinMod{p_+}{p_-}\) is completely
reducible and that it can be used to construct rational conformal field
theories requires only a little more work. 
The Virasoro Verma module of conformal weight \(h_{r,s}\), where \(1\leq r\leq p_+-1\) and
\(1\leq s\leq p_--1\), contains a maximal
subrepresentation generated by two independent singular vectors of conformal weights
\(h^\prime=h_{r,s}+rs\) and
\(h^{\prime\prime}=h_{r,s}+(p_+-r)(p_--s)\) \cite{FFFock90}. However, neither
\(h^{\prime}\) nor \(h^{\prime\prime}\) are roots of \eqref{eq:minpoly}.
So the \(\MinMod{p_+}{p_-}\) representation of conformal weight \(h_{r,s}\) must
be isomorphic to
the irreducible quotient of the　Virasoro Verma module of conformal weight \(h_{r,s}\).
This also implies that there exists no non-trivial extensions between 
irreducible representations with distinct conformal weights. In \cite[Prop.~7.5]{RidSta09}, it was shown that
the irreducible Virasoro representation of conformal weight \(h_{r,s}\) admits
no self extensions (as representations of the Virasoro algebra).
Thus, neither do the irreducible \(\MinMod{p_+}{p_-}\) representations.
This proves that irreducible \(\MinMod{p_+}{p_-}\) representations
do not admit any non-trivial extensions and that therefore the representation theory of \(\MinMod{p_+}{p_-}\) is
completely reducible.

\begin{cor}
The Virasoro minimal model \voas{} are rational, that is, they admit only a finite number of inequivalent irreducible representations and all representations are completely reducible.
\end{cor}


\end{document}